\documentclass[conference]{IEEEtran}

\IEEEoverridecommandlockouts
\usepackage{graphicx}
\usepackage{amsmath}
\usepackage[noend]{algpseudocode}
\usepackage{algorithmicx,algorithm}
\usepackage{amsthm}
\usepackage{amsfonts}
\usepackage{subfigure} 
\usepackage{color}
\usepackage{cite}
\usepackage{setspace}
\newtheorem{Lemma}{Lemma}
\newtheorem{Theorem}{Theorem}

\newtheorem{Proposition}{Proposition}
\usepackage{amssymb}
\usepackage{color}
\graphicspath{{images/}}
\usepackage{booktabs}
\usepackage{multirow} 
\usepackage{makecell}
\usepackage[numbers,sort&compress]{natbib}

\usepackage{stfloats}
\usepackage{amsmath}
\usepackage{amssymb}
\usepackage{enumitem}
\usepackage{threeparttable}
\usepackage{comment}
\usepackage{bm}
\usepackage{colortbl}
\usepackage{tabularx}
\usepackage[table]{xcolor}
\begin{document}	
	\title{Fluid-Dynamic Interference Modeling for LEO Mega-Constellations: A Spatiotemporal Kinetic Field Approach}
	\author{Wen-Yu Dong, Weiwei Jiang,~\IEEEmembership{Senior Member,~IEEE}, Song Zhao, Rui-Si Han,\\ Qi Bi,~\IEEEmembership{Fellow,~IEEE}, Sheng Chen,~\IEEEmembership{Life Fellow,~IEEE}	%
		
		\thanks{W.-Y. Dong, S. Zhao and Q. Bi are with Future Technology Research Center, China Telecom Research Institute, Beijing 102209, China (E-mails: dongwy@chinatelecom.cn; zhaosong1@chinatelecom.cn; qibi@chinatelecom.cn).} %
		\thanks{W. Jiang is with the School of Information and Communication Engineering, Beijing University of Posts and Telecommunications, Beijing, 100876, China (Email: jww@bupt.edu.cn)}
		\thanks{R.-S.  Han is with Cloud Network Operating System R\&D Center, China Telecom, Beijing 102209, China (E-mail: hanruisi@chinatelecom.cn).} %
		\thanks{S. Chen is with the School of Electronics and Computer Science, University of Southampton, Southampton SO17 1BJ, U.K. (E-mail: sqc@ecs.soton.ac.uk).} %
		\vspace*{-8mm}
	}
	
	
	
	\maketitle 
	\vspace{-1.5cm}
	\pagestyle{empty}
	\thispagestyle{empty}
	\begin{abstract}
		Low Earth orbit (LEO) mega-constellations create a highly non-stationary interference environment that cannot be accurately captured by static stochastic-geometry snapshots. This paper proposes a kinetic interference field framework that models the constellation as a compressible fluid shell evolving under orbital kinematics. By mapping satellite motion into a continuum flux field, we derive a hydrodynamic conservation law for the aggregate interference and obtain a closed-form expression for the time-varying outage probability via moment matching. The analysis reveals that high-latitude ``interference surges'' are a direct consequence of orbital compression and boundary flux, rather than random anomalies. Numerical validation against ephemeris-driven Monte Carlo simulations confirms the accuracy of the framework across time evolution, latitude, and design parameters. Leveraging the closed-form model, we further show that the conventional $90^{\circ}$ polar-orbit design is not universally outage-optimal. Instead, an inclination angle near $79^{\circ}$ at low altitude achieves a favorable trade-off between coverage continuity and geometric interference isolation. The proposed framework provides a tractable analytical tool for interference-aware 6G non-terrestrial network (NTN) design.
	\end{abstract}
	
	\vspace*{-2mm}
	\section{Introduction}\label{S1}
	
	\IEEEPARstart{T}{he} integration of low Earth orbit (LEO) mega-constellations into the 6G ecosystem represents a major transition from static terrestrial topologies to highly dynamic non-terrestrial networks (NTNs) \cite{Yaacoub2020}. Unlike geostationary systems, LEO constellations consist of dense shells of satellites orbiting at high velocities. While these architectures promise seamless global coverage, the requisite aggressive frequency reuse creates an analytically challenging, interference-limited environment with rapid spatiotemporal evolution.
	A fundamental analytical bottleneck lies in the intrinsic non-stationarity of the LEO interference field. Unlike terrestrial networks modeled by stationary point processes \cite{Andrews2011,Haenggi2012}, LEO networks are governed by deterministic Keplerian mechanics, which enforce rigorous conservation of angular momentum. This orbital physics manifests as a ``polar compression'' phenomenon---or orbital caustics---where satellite density diverges near the inclination limits and becomes sparse near the equator \cite{Walker1984}. Consequently, the aggregate interference is not merely a random spatial variable but a structured spatiotemporal process driven by orbital convergence.
	
	Conventional analyses face a persistent trade-off between scalability and kinematic consistency. Studies that predominantly rely on static stochastic geometry (SG) snapshots \cite{DongTCOM, DongJSAC, IoTJDong} fundamentally decouple the temporal correlation of interference from the underlying kinematics. Such static approximations often struggle to capture the duration and intensity of interference surges. Conversely, while discrete-event simulations offer geometric fidelity, they lack analytical tractability and suffer from severe computational scalability bottlenecks for emerging mega-constellations \cite{Portillo2019}. Furthermore, traditional stochastic mobility models designed for terrestrial networks are not well suited to capturing the deterministic nature of Keplerian orbits. Recent fluid-spatiotemporal frameworks have demonstrated the value of density--flux representations for modeling non-stationary information transport and extending continuum-field methods to network provisioning and routing optimization \cite{Dong2026FSTSG,Dong2026TMC,Dong2026VDR}. This methodological gap motivates their further development for dynamic interference analysis in LEO mega-constellations.  
	
	In this paper, we propose a kinetic interference field (KIF) framework, which conceptualizes the mega-constellation not as a set of discrete nodes, but as a compressible fluid shell evolving under hydrodynamic conservation laws. By mapping discrete orbital kinematics to continuous fluid flux \cite{Lighthill1955, Lasry2007}, we transform the stochastic interference summation into a deterministic partial differential equation (PDE). This framework bridges the gap between the discrete granularity of orbital dynamics and the analytical tractability required for system-level optimization.
	The primary contributions of this paper are summarized as follows.
	
	\begin{itemize}
		\item \textit{Kinetic Field Framework:} We establish a continuum theory mapping discrete Walker constellation kinematics to a macroscopic fluid flow. This formulation provides a tractable mean-field representation of the aggregate interference over the joint probability space of traffic and beam orientations. By transforming the discrete interference summation into a fixed-resolution numerical quadrature, its evaluation avoids explicit scaling with the number of satellites, alleviating the scalability bottleneck of ephemeris tracking.
		\item \textit{Hydrodynamic PDE Derivation:} We derive a governing PDE that characterizes the spatiotemporal evolution of the aggregate interference field. By decomposing the non-stationary fluctuations into geometric advection, traffic gradient drift and boundary flux, this formulation analytically shows that high-latitude interference surges are deterministic consequences of orbital flux convergence.
		\item \textit{Dynamic Reliability Analysis:} To bridge the macroscopic fluid flow and the microscopic discrete network, we formulate closed-form expressions for the time-varying outage probability (TVOP). By employing a moment-matching Gamma approximation, we characterize the second-order spatial shot noise of the thinned satellite process and obtain a tractable metric for dynamic reliability evaluation.
		\item \textit{Topology Optimization:} Leveraging the analytical tractability of the framework, we evaluate the altitude--inclination design space beyond the conventional strict polar-orbit configuration. Under the evaluated system parameters, an inclination angle of $\iota \approx 79^{\circ}$ at low altitudes provides a favorable trade-off between global coverage and geometric interference isolation.
	\end{itemize}
	
	\vspace*{-2mm}
	\section{System Model and Kinetic Formulation}\label{S2}
	\begin{figure}[!t]
		\begin{center}
			\includegraphics[width=0.96\columnwidth]{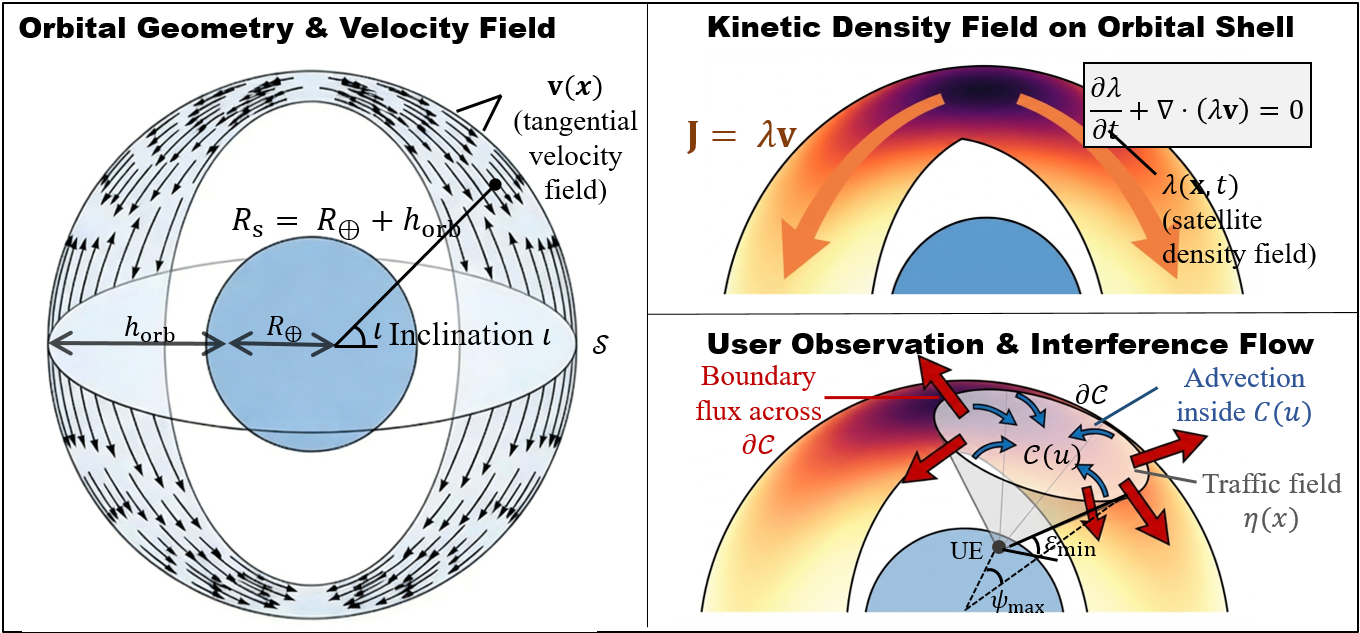}
		\end{center}
		\vspace*{-5mm}
		\caption{Geometric representation of the proposed KIF framework.}
		\label{fig:system_model} 
		\vspace*{-5mm}
	\end{figure}
	
	As illustrated in Fig. \ref{fig:system_model}, we consider a downlink LEO satellite communication network. The Earth is modeled as a sphere of radius $R_{\oplus}$. Satellites are distributed on a concentric orbital shell $\mathcal{S}$ of radius $R_{\mathrm{s}} = R_{\oplus} + h_{\mathrm{orb}}$, where $h_{\mathrm{orb}}$ is the orbital altitude.
	
	\vspace*{-2mm}
	\subsection{Orbital Kinematics}
	
	We consider a Walker delta constellation consisting of $N$ satellites with a common inclination angle $\iota$. Assuming circular orbits, the satellite maintains a constant linear orbital speed $v_{\mathrm{orb}} = \sqrt{\mu / R_{\mathrm{s}}}$, where $\mu$ represents the Earth standard gravitational parameter. The velocity field $\mathbf{v}(\mathbf{x})$ at a position $\mathbf{x}$ parameterized by latitude $\phi$ decomposes into latitudinal $v_{\phi}$ and longitudinal $v_{\theta}$ components:
	\begin{equation} 
		\label{eq:velocity_components}
		v_{\phi}(\phi) = v_{\mathrm{orb}} \sqrt{1 - \frac{\cos^2 \iota}{\cos^2 \phi}}, \quad v_{\theta}(\phi) = v_{\mathrm{orb}} \frac{\cos \iota}{\cos \phi},
	\end{equation}
	defined for the region $|\phi| \le \iota$. Equation \eqref{eq:velocity_components} indicates that the latitudinal velocity $v_{\phi}$ vanishes as the satellite approaches the inclination limits $|\phi| \to \iota$. This kinematic stagnation drives the density singularity at the poles, known as orbital caustics.
	
	\vspace*{-2mm}
	\subsection{Continuum Field Representation}
	
	To characterize the spatiotemporal dynamics of the mega-constellation, we transition from a discrete microscopic state to a macroscopic continuum representation. The discrete satellite distribution is replaced by a continuous satellite-density field $\lambda(\mathbf{x}, t)$. The spatiotemporal evolution of the satellite density field on the shell $\mathcal{S}$ is governed by the continuity equation:
	\begin{equation} 
		\label{eq:continuity}
		\frac{\partial \lambda(\mathbf{x}, t)}{\partial t} + \nabla \cdot (\lambda(\mathbf{x}, t) \mathbf{v}(\mathbf{x})) = 0,
	\end{equation}
	where $\nabla \cdot$ represents the surface divergence operator on the shell $\mathcal{S}$. For a Walker constellation with uniform orbital spacing, the static density profile $\bar{\lambda}(\phi)$ is derived to ensure that the flux is conserved across latitudes:
	\begin{equation} 
		\label{eq:static_density}
		\bar{\lambda}(\phi) = \frac{N}{2\pi^2 R_{\mathrm{s}}^2 \sqrt{\sin^2 \iota - \sin^2 \phi}}, \quad \text{for } |\phi| < \iota.
	\end{equation}
	
	\vspace*{-2mm}
	\subsection{Channel Model and User Association}
	
	A reference ground user equipment (UE) is located at $\mathbf{u}$ with $\|\mathbf{u}\| = R_{\oplus}$. A satellite at position $\mathbf{x} \in \mathcal{S}$ is visible to the user if its elevation angle exceeds a minimum threshold $\varepsilon_{\min}$. The set of visible satellites forms a spherical cap $\mathcal{C}(\mathbf{u})$ on the shell $\mathcal{S}$. The user selects the serving satellite $\mathbf{x}_0 \in \mathcal{C}(\mathbf{u})$ based on the maximum received signal power criterion. The received signal power is given by:
	\begin{equation}\label{eq:signal} 
		S(t) = P_{\mathrm{tx}} G_{\mathrm{tx}}(\theta_{0}) G_{\mathrm{rx}}(0) \|\mathbf{x}_0 - \mathbf{u}\|^{-\alpha} H_0,
	\end{equation}
	where $P_{\mathrm{tx}}$ denotes the transmit power, $\alpha$ is the path loss exponent, and $H_0$ models the small-scale fading channel gain, while $G_{\mathrm{tx}}(\cdot)$ and $G_{\mathrm{rx}}(\cdot)$ denote the transmit and receive antenna gain patterns, respectively, with $\theta_{0}$ representing the instantaneous off-axis angle of the transmit beam relative to the user direction. 
	
	\vspace*{-2mm}
	\subsection{Aggregate Interference Formulation}
	
	All visible satellites within the spherical cap $\mathcal{C}(\mathbf{u}) \setminus \{\mathbf{x}_0\}$ constitute the set of potential interferers. We introduce a location-dependent traffic activation field $\eta(\mathbf{x}) \in [0, 1]$, which quantifies the transmission probability of a satellite at location $\mathbf{x}$. The random antenna gain is approximated by the effective interference gain $\bar{G}_{\mathrm{tx}}$, defined as the statistical expectation of the radiated gain over the satellite coverage footprint.
	
	Within the fluid dynamic framework, the aggregate interference $I(\mathbf{u}, t)$ is obtained by integrating these continuous field components over the visibility domain:
	\begin{equation} 
		\label{eq:fluid_interference}
		I(\mathbf{u}, t)\! =\! \int_{\mathcal{C}(\mathbf{u})} \! \underbrace{\eta(\mathbf{x}) P_{\mathrm{tx}} \bar{G}_{\mathrm{eq}}(\mathbf{u}, \mathbf{x}) \|\mathbf{x} - \mathbf{u}\|^{-\alpha}}_{\text{Effective kernel } \tilde{g}(\mathbf{u}, \mathbf{x})} \lambda(\mathbf{x}, t)  \mathrm{d}\mathbf{x},
	\end{equation}
	where $\bar{G}_{\mathrm{eq}}(\mathbf{u}, \mathbf{x}) = \bar{G}_{\mathrm{tx}} G_{\mathrm{rx}}(\theta_{\mathrm{rx}})$ encapsulates the expected combined antenna gain, with $\theta_{\mathrm{rx}}$ denoting the off axis angle of the interfering signal relative to the user receiving boresight. The effective kernel $\tilde{g}(\mathbf{u}, \mathbf{x})$ unifies path loss, mean-field beam directionality, and local traffic probability into a deterministic weighting function.
	
	\vspace*{-2mm}
	\section{Fluid Dynamic Interference Analysis}\label{S3}
	
	In this section, we derive the governing equations of the KIF. Leveraging the continuum representation, we characterize the spatiotemporal evolution of the aggregate interference through a PDE driven by the orbital velocity field.
	
	\vspace*{-2mm}
	\subsection{Dynamics of Aggregate Interference}
	
	The instantaneous aggregate interference power $I(\mathbf{u}, t)$ is an integral functional of the satellite density field $\lambda(\mathbf{x}, t)$ weighted by the effective kernel $\tilde{g}(\mathbf{u}, \mathbf{x})$. We explicitly decompose the kernel as $\tilde{g}(\mathbf{u}, \mathbf{x}) = \eta(\mathbf{x}) g(\mathbf{u}, \mathbf{x})$, where $\eta(\mathbf{x})$ denotes the spatial traffic probability field and $g(\mathbf{u}, \mathbf{x}) = P_{\mathrm{tx}} \bar{G}_{\mathrm{eq}}(\mathbf{u}, \mathbf{x}) \|\mathbf{x} - \mathbf{u}\|^{-\alpha}$ represents the deterministic channel gain. In the Earth-fixed frame, the user location $\mathbf{u}$ is static, rendering the integration domain $\mathcal{C}(\mathbf{u})$ time-invariant. Consequently, the temporal evolution of $I(\mathbf{u}, t)$ is governed solely by the hydrodynamic flow of the active density field $\lambda(\mathbf{x}, t)$ passing through this fixed observation window.
	
	\begin{Proposition}
		\label{prop:interference_pde}
		The temporal rate of change of the aggregate interference $I(\mathbf{u}, t)$ is governed by the superposition of orbital advection, traffic gradient drift, and boundary flux:
		\begin{equation} 
			\label{eq:interference_evolution}
			\frac{\partial I(\mathbf{u}, t)}{\partial t} = \mathcal{T}_{\mathrm{adv}}^{\mathrm{geo}} + \mathcal{T}_{\mathrm{adv}}^{\mathrm{traf}} - \mathcal{T}_{\mathrm{flux}},
		\end{equation}
		where the constituent terms are defined as:
		\begin{align} 
			\mathcal{T}_{\mathrm{adv}}^{\mathrm{geo}} &= \int_{\mathcal{C}(\mathbf{u})} \eta(\mathbf{x}) \lambda(\mathbf{x}, t) \mathbf{v}(\mathbf{x}) \cdot \nabla g(\mathbf{u}, \mathbf{x}) \, \mathrm{d}\mathbf{x}, \label{eq:term_geo} \\
			\mathcal{T}_{\mathrm{adv}}^{\mathrm{traf}} &= \int_{\mathcal{C}(\mathbf{u})} g(\mathbf{u}, \mathbf{x}) \lambda(\mathbf{x}, t) \mathbf{v}(\mathbf{x}) \cdot \nabla \eta(\mathbf{x}) \, \mathrm{d}\mathbf{x}, \label{eq:term_traffic} \\
			\mathcal{T}_{\mathrm{flux}} &= \oint_{\partial \mathcal{C}(\mathbf{u})} \eta(\mathbf{x}) g(\mathbf{u}, \mathbf{x}) \lambda(\mathbf{x}, t) \mathbf{v}(\mathbf{x}) \cdot \mathbf{n}(\mathbf{x}) \, \mathrm{d}l. \label{eq:term_flux}
		\end{align}
		Here, $\nabla$ denotes the surface gradient operator on the orbital shell, $\mathbf{n}(\mathbf{x})$ is the outward unit normal vector on the boundary $\partial \mathcal{C}(\mathbf{u})$, and $\mathrm{d}l$ is the differential arc length.
	\end{Proposition}
	
	Proposition \ref{prop:interference_pde} provides a physical interpretation of the non-stationary interference environment in LEO networks. The advection dynamics are decomposed into two physical components. The geometric drift $\mathcal{T}_{\mathrm{adv}}^{\mathrm{geo}}$ characterizes the variation in the physical link budget caused by satellite motion, weighted by the local traffic probability $\eta(\mathbf{x})$. It accounts for the temporal rate of change in path loss and the mean-field beam scanning effect. The traffic gradient drift $\mathcal{T}_{\mathrm{adv}}^{\mathrm{traf}}$ represents the interference fluctuation driven by satellites crossing the boundaries of traffic demand zones. The convective derivative $\mathbf{v} \cdot \nabla \eta$ functions as a distributed source or sink term for the KIF. It introduces new interference power when satellites move from a dormant region into an active service area. Conversely, it generates a dissipative flux when active satellites cease transmission, resolving the apparent discontinuity in the interference field.
	The boundary flux $\mathcal{T}_{\mathrm{flux}}$ quantifies the net flow of aggregate interference across the visibility horizon. In high-latitude regions, orbital convergence leads to a rapid interference surge. A critical mathematical concern arises at the inclination limits ($\phi \to \pm\iota$), where the latitudinal velocity vanishes ($v_{\phi} \to 0$), causing the kinetic density $\lambda$ to diverge as a singularity of the form $1/\sqrt{\sin^2 \iota - \sin^2 \phi}$. 
	
	\begin{Proposition}
		\label{prop:boundedness}
		Despite the kinetic density diverging at the orbital caustics ($\phi \to \pm\iota$), the aggregate interference power $I(\mathbf{u}, t)$ is strictly bounded for any non-zero orbital altitude $h_{\mathrm{orb}} > 0$.
	\end{Proposition}
	\begin{proof}
		The effective kernel is bounded by a finite constant $M_{\max} = P_{\mathrm{tx}} G_{\max}^2 h_{\mathrm{orb}}^{-\alpha} < \infty$. The interference integral is bounded by $I \le M_{\max} \int_{\mathcal{S}} \lambda(\mathbf{x}, t) \mathrm{d}\mathbf{x}$. The spatial integration over the density field contains a singularity at $\phi \to \pm\iota$. Using the variable substitution $u = \sin \phi / \sin \iota$, the probability measure evaluates analytically to $\int_{-\iota}^{\iota} (\sin^2 \iota - \sin^2 \phi)^{-1/2} \cos \phi \mathrm{d}\phi = \int_{-1}^{1} (1-u^2)^{-1/2} \mathrm{d}u = \pi$. Thus, from a real analysis perspective, the singularity is strictly Lebesgue integrable (order $-1/2$), ensuring $\sup_{\mathbf{u}, t} I(\mathbf{u}, t) < \infty$.
	\end{proof}
	
	This mathematical proof establishes that high-latitude interference surges are deterministic physical power concentrations rather than non-physical mathematical divergences.
	
	\vspace*{-2mm}
	\subsection{Evolution of Interference Statistics via Stochastic-Kinetic Closure}
	
	While Equation \eqref{eq:interference_evolution} governs the macroscopic deterministic flow of the aggregate interference power, the discrete nature of finite-sized constellations and the random spatial traffic requests introduce intrinsic spatial shot noise. To rigorously bridge deterministic orbital mechanics and stochastic interference fluctuations, we introduce a stochastic-kinetic closure grounded in point process limit theorems.
	
	We emphasize that the underlying Walker constellation forms a deterministic geometric lattice characterized by strict spatial regularity. However, the set of active interfering satellites is a randomly thinned subset of this lattice, driven by the independent spatial traffic activation field $\eta(\mathbf{x})$. Based on the independent thinning theorem for point processes, as the constellation becomes dense, the randomly thinned regular lattice converges weakly to an inhomogeneous Poisson point process (PPP). 
	For finite $N$, approximating the repulsive discrete lattice with a completely random PPP structurally ignores the inter-satellite minimum distance constraints (i.e., the hard-core nature). By the principles of stochastic ordering, removing this spatial regularity strictly overestimates the interference variance. Therefore, this PPP closure is not merely a heuristic approximation, but provides a mathematically guaranteed conservative (worst-case) upper bound for the outage probability.
	
	To formalize this reliability guarantee, we establish the following lemma regarding the stochastic ordering of the interference variance.
	
	\begin{Lemma}
		\label{lemma:ppp_bound}
		Let $I_{\mathrm{grid}}(\mathbf{u}, t)$ denote the true aggregate interference from the randomly thinned deterministic Walker lattice, and $I_{\mathrm{ppp}}(\mathbf{u}, t)$ denote the interference from the surrogate inhomogeneous PPP with an identical intensity measure $\Lambda(\mathbf{x}, t) = \eta(\mathbf{x}) \lambda(\mathbf{x}, t)$. Due to the strictly repulsive nature (hard-core constraint) of the orbital grid, the interference variance satisfies $\mathrm{Var}[I_{\mathrm{grid}}] \le \mathrm{Var}[I_{\mathrm{ppp}}]$. Consequently, under the adopted Gamma moment-matching closure, the PPP surrogate provides a conservative approximation to the system outage: $\mathbb{P}[\gamma_{\mathrm{grid}} < \gamma_{\mathrm{th}}] \le \mathbb{P}[\gamma_{\mathrm{ppp}} < \gamma_{\mathrm{th}}]$.
	\end{Lemma}
	\begin{proof}
		The discrete Walker constellation enforces a minimum inter-satellite angular separation, fundamentally constituting a repulsive point process. According to the void probability and second-order factorial moment properties of repulsive processes \cite{Baccelli2009}, replacing the hard-core lattice with a completely spatially random PPP (which permits arbitrarily close satellite pairs) strictly amplifies the spatial variance of the interference field. Therefore, substituting $\sigma_{\mathrm{I}}^2$ evaluated under the PPP assumption into the monotonically increasing outage functional guarantees a worst-case performance bound.
	\end{proof}
	
	Under this rigorous closure formulation, the time-varying mean $\mu_{\mathrm{I}}(t) = \mathbb{E}[I(\mathbf{u}, t)]$ is precisely governed by the linear transport equation derived in Proposition \ref{prop:interference_pde}:
	\begin{equation} 
		\frac{\mathrm{d} \mu_{\mathrm{I}}}{\mathrm{d}t} = \mathcal{T}_{\mathrm{adv}} - \mathcal{T}_{\mathrm{flux}} =
		\mathcal{T}_{\mathrm{adv}}^{\mathrm{geo}} + \mathcal{T}_{\mathrm{adv}}^{\mathrm{traf}} - \mathcal{T}_{\mathrm{flux}} .
	\end{equation}
	
	The interference variance $\sigma_{\mathrm{I}}^2(t) = \mathrm{Var}[I(\mathbf{u}, t)]$ is determined by the second-order moment measure of the PPP. It is expressed as the integral of the squared response kernel weighted by the active density. Differentiating the variance with respect to time and invoking the Reynolds transport theorem yields a conservation law for the variance:
	\begin{equation} 
		\label{eq:variance_evolution}
		\frac{\mathrm{d} \sigma_{\mathrm{I}}^2}{\mathrm{d}t} = \int_{\mathcal{C}(\mathbf{u})} \lambda \mathbf{v} \cdot \nabla (\eta g^2) \, \mathrm{d}\mathbf{x} - \oint_{\partial \mathcal{C}(\mathbf{u})} \eta g^2 \lambda \mathbf{v} \cdot \mathbf{n} \, \mathrm{d}l.
	\end{equation}
	
	The variance evolution is governed by the interplay of geometric scanning and traffic gradient drift. In regions of orbital convergence, both the average interference level and the dispersion of the interference intensify. This second-order dynamic is critical for accurately modeling the tail behavior of the interference distribution during outage analysis.
	
	\vspace*{-1mm}
	\section{Dynamic Outage Probability Analysis}\label{S4}
	
	Having established the deterministic evolution laws for the mean and variance of the aggregate interference, we now derive the  time-varying outage probability (TVOP). This metric quantifies the instantaneous reliability of the link conditioned on the dynamic state of the KIF.
	
	\vspace*{-2mm}
	\subsection{Statistical Characterization via Moment Closure}
	
	To analytically evaluate the outage probability without relying on numerical inversions of the characteristic function, we employ a mathematically tractable moment-matching closure. We approximate the probability density function of the aggregate interference as a Gamma distribution, $I \sim \text{Gamma}(k_{\mathrm{I}}(t), \theta_{\mathrm{I}}(t))$, whose parameters are uniquely determined by the instantaneous kinetic fluid moments:
	\begin{equation} 
		\label{eq:moment_matching}
		k_{\mathrm{I}}(t) = \frac{\mu_{\mathrm{I}}^2(t)}{\sigma_{\mathrm{I}}^2(t)}, \quad \theta_{\mathrm{I}}(t) = \frac{\sigma_{\mathrm{I}}^2(t)}{\mu_{\mathrm{I}}(t)}.
	\end{equation}
	
	This moment-closure approach is a practical engineering approximation suitable for dense, interference-limited regimes where the aggregate interference consists of numerous contributions. It translates the macroscopic deterministic flow ($\mu_{\mathrm{I}}$) and spatial dispersion ($\sigma_{\mathrm{I}}^2$) into a tractable dynamic reliability metric.
	
	\vspace*{-2mm}
	\subsection{Time-Varying Outage Probability}
	
	Conditioning on the aggregate interference, we obtain the closed-form expression for the outage probability $P_{\mathrm{out}}(t) = \mathbb{P}[\gamma(\mathbf{u}, t) < \gamma_{\mathrm{th}}]$, where $\gamma_{\mathrm{th}}$ is the predefined SINR threshold.
	
	\begin{Theorem}
		\label{thm:outage_probability}
		Consider a user served by a Nakagami-$m$ fading signal in the presence of Gamma-distributed fluid interference. The instantaneous outage probability $P_{\mathrm{out}}(t)$ is given by:
		\begin{align} 
			\label{eq:closed_form_outage}
			P_{\mathrm{out}}(t) = & 1 - \sum_{n=0}^{m_{\mathrm{S}}-1} \sum_{j=0}^{n} \Xi_{n,j}(t) \cdot \Gamma(k_{\mathrm{I}}(t)+j) \nonumber\\
			&\times  \left( \frac{1}{\theta_{\mathrm{I}}(t)} + \frac{\gamma_{\mathrm{th}}}{\theta_{\mathrm{S}}(t)} \right)^{-(k_{\mathrm{I}}(t)+j)},
		\end{align}
		where the time-varying coefficient $\Xi_{n,j}(t)$ is defined as:
		\begin{equation} 
			\Xi_{n,j}(t) = \frac{\binom{n}{j} (\gamma_{\mathrm{th}}\sigma^2)^{n-j} \gamma_{\mathrm{th}}^j}{n! \, \theta_{\mathrm{S}}(t)^n \Gamma(k_{\mathrm{I}}(t)) \theta_{\mathrm{I}}(t)^{k_{\mathrm{I}}(t)}} \exp\left(-\frac{\gamma_{\mathrm{th}}\sigma^2}{\theta_{\mathrm{S}}(t)}\right).
		\end{equation}
		Here $\theta_{\mathrm{S}}(t)\! =\! \bar{S}(t) / m_{\mathrm{S}}$ is the time-varying scale parameter of the Gamma-distributed desired signal power, and $\bar{S}(t)$ is the local mean received signal power excluding small-scale fading.
	\end{Theorem}
	
	Theorem \ref{thm:outage_probability} provides an explicit analytical mapping from the fluid state moments to the system reliability, enabling the rapid evaluation of outage dynamics without the computational burden of extensive Monte Carlo simulations.
	
	\vspace*{-2mm}
	\subsection{Differential Sensitivity Analysis}
	
	To elucidate the physical drivers governing outage events, we analyze the temporal derivative of $P_{\mathrm{out}}(t)$. Using the chain rule, its temporal evolution is directly linked to the fluid moment state vector:
	\begin{equation} 
		\label{eq:outage_dynamics}
		\frac{\mathrm{d} P_{\mathrm{out}}}{\mathrm{d}t} = \underbrace{\mathbf{w}^T \boldsymbol{\mathcal{T}}_{\mathrm{adv}}}_{\text{Geometry Drift}} - \underbrace{\mathbf{w}^T \boldsymbol{\mathcal{T}}_{\mathrm{flux}}}_{\text{Boundary Flux}},
	\end{equation}
	where $\mathbf{w} = \big[\frac{\partial P_{\mathrm{out}}}{\partial \mu_{\mathrm{I}}}, \frac{\partial P_{\mathrm{out}}}{\partial \sigma_{\mathrm{I}}^2}\big]^T$ denotes the sensitivity vector, and $\boldsymbol{\mathcal{T}}_{\mathrm{adv}}$ and $\boldsymbol{\mathcal{T}}_{\mathrm{flux}}$ represent the corresponding $2 \times 1$ advection and boundary flux column vectors for the fluid moments. In high-latitude regions, the boundary flux term becomes the dominant driver of outage surges. This confirms that the rapid deterioration of SINR is not merely a geometric coincidence, but a deterministic consequence of the hydrodynamic flux crossing the visibility horizon.
	
	\begin{table}[t]
		\vspace*{-1mm}
		\scriptsize
		\caption{Default Simulation Parameters}
		\vspace*{-3mm}
		\label{tab:sim_params}
		\centering
		\renewcommand{\arraystretch}{1.1}
		\begin{tabular}{l c c}
			\toprule
			Parameter & Symbol & Value \\
			\midrule
			Orbital altitude & $h_{\mathrm{orb}}$ & 1000 km  \\
			Inclination angle & $\iota$ & $85^{\circ}$ \\
			Total satellites & $N$ & 960 \\
			Orbital velocity & $v_{\mathrm{orb}}$ & 7.35 km/s \\
			Carrier frequency & $f_{\mathrm{c}}$ & 28 GHz \\
			Transmit power & $P_{\mathrm{tx}}$ & 30 dBm \\
			Path loss exponent & $\alpha$ & 3.5 \\
			Noise power & $\sigma^2$ & $-95$ dBm \\
			Peak antenna gain & $G_{\max}$ & 32 dBi \\
			Side lobe attenuation & $\Delta G$ & 18 dB \\
			Small-scale fading & $m_{\mathrm{S}}$ & 2 \\
			SINR threshold & $\gamma_{\mathrm{th}}$ & 0 dB \\
			Min elevation angle & $\varepsilon_{\min}$ & $10^{\circ}$ \\
			\bottomrule
		\end{tabular}
		\vspace*{-3mm}
	\end{table}
	
	\vspace*{-2mm}
	\section{Numerical Results}\label{S5}
	
	This section validates the proposed KIF framework and derives system-level design insights. The theoretical predictions are benchmarked against ephemeris-driven Monte Carlo simulations, which instantiate discrete Walker delta constellations based on Keplerian mechanics. To ensure statistical convergence, empirical benchmarks represent the ensemble average of 50,000 independent snapshots.The default simulation parameters are listed in Table \ref{tab:sim_params}.
	
	\begin{figure}[h]
		\vspace*{-1mm}
		\centering
		\includegraphics[width=0.9\linewidth]{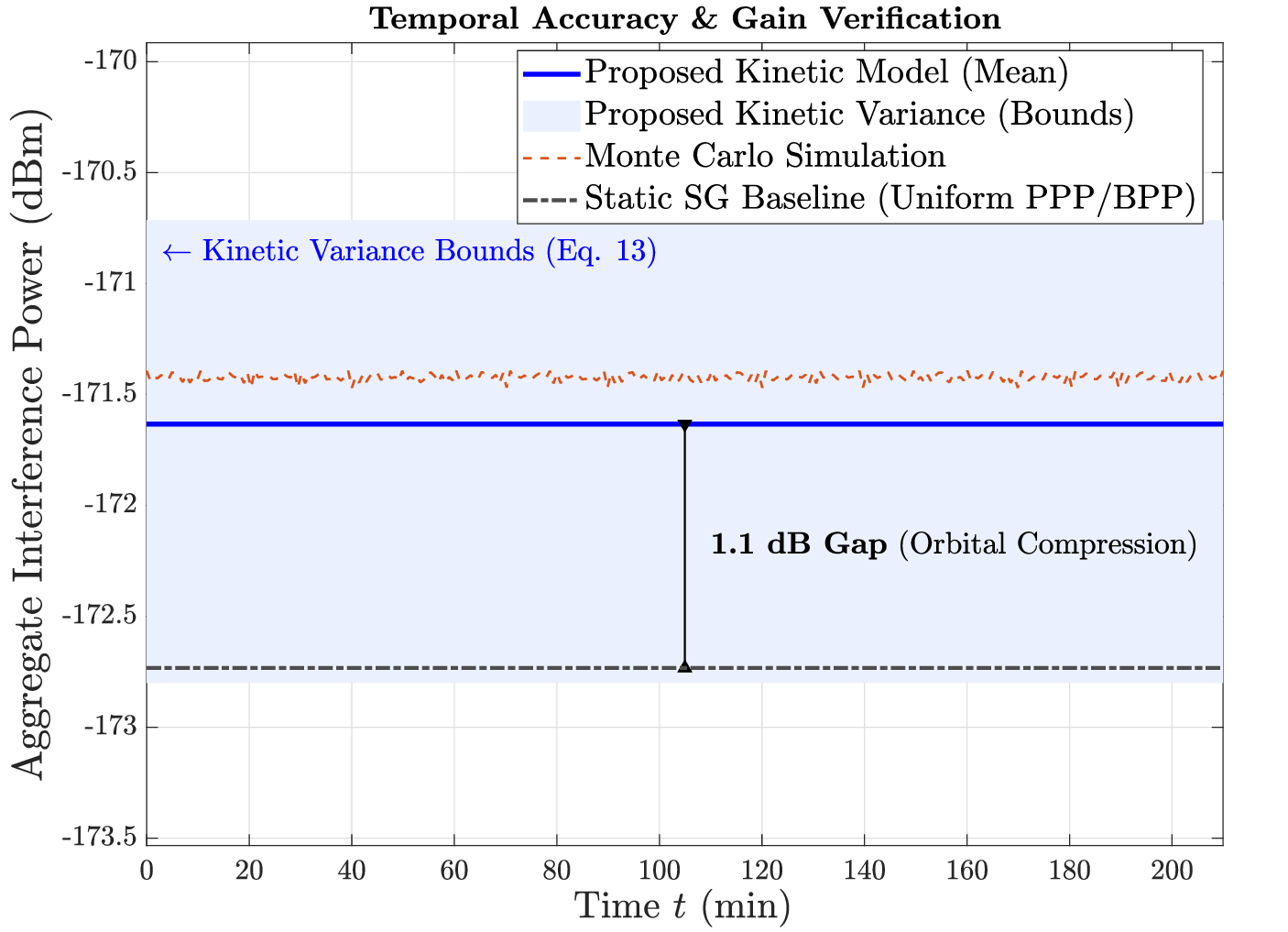} 
		\vspace*{-4mm}
		\caption{Temporal evolution of aggregate interference power.}
		\label{fig:temporal_accuracy}
		\vspace*{-4mm}
	\end{figure}
	
	\begin{figure}[b]
		\vspace*{-5mm}
		\centering
		\includegraphics[width=0.9\linewidth]{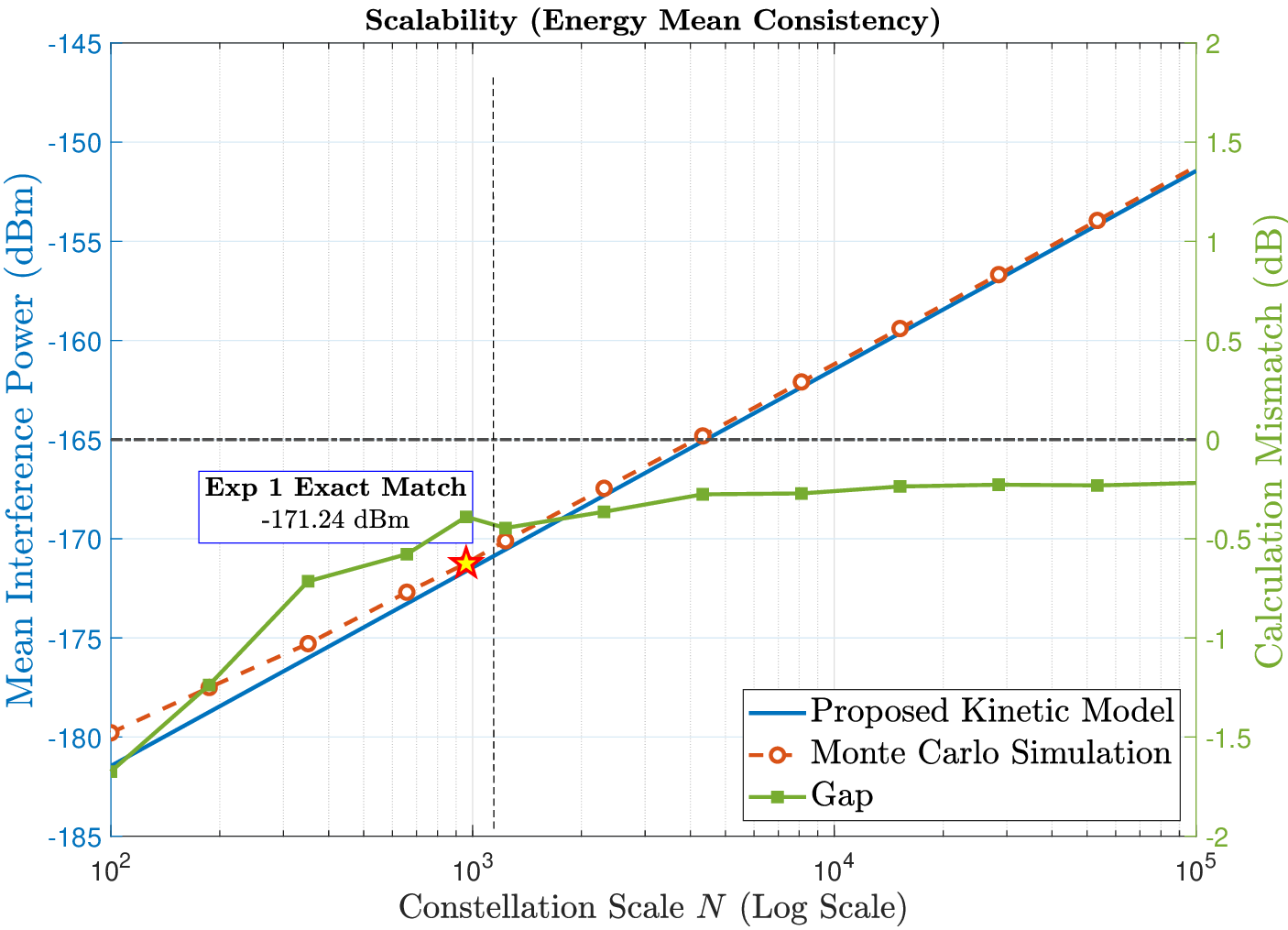} 
		\vspace*{-4mm}
		\caption{Scalability verification across constellation scales.}
		\vspace*{-1mm}
		\label{fig:scalability}
	\end{figure}
	
	\vspace*{-2mm}
	\subsection{Validation of Kinetic Interference Dynamics}
	
	We first validate the kinetic interference evolution laws against Monte Carlo simulations. Fig. \ref{fig:temporal_accuracy} illustrates the temporal evolution of the aggregate interference power. The kinetic mean accurately tracks the empirical mean, confirming that the macroscopic fluid representation captures the first-order statistical behavior. The predicted variance also captures the scale of the microscopic stochastic fluctuations. Furthermore, we benchmark our framework against a conventional homogeneous static SG model. As shown in Fig. \ref{fig:temporal_accuracy}, the static baseline systematically underestimates the aggregate interference by approximately 1.1 dB. This deviation highlights the fundamental limitation of assuming a uniform spatial distribution. By uniquely integrating the velocity field, the proposed KIF framework analytically predicts the continuous interference surges driven by the orbital compression effect, entirely bypassing the computationally prohibitive need for snapshot-by-snapshot spatial refitting.

	To verify the framework's applicability across diverse network scales, Fig. \ref{fig:scalability} analyzes the mean aggregate interference power as a function of the constellation size $N$. As the constellation density increases, the modeling error decreases. This confirms that the discrete Walker geometry approaches the continuum regime, validating the continuum hypothesis for massive networks. 
	
	\begin{figure}[h]
		\vspace*{-4.5mm}
		\centering
		\subfigure[{\scriptsize Latitudinal density profile}] {
			\includegraphics[width=0.8\linewidth,trim=0 3 0 0, clip]{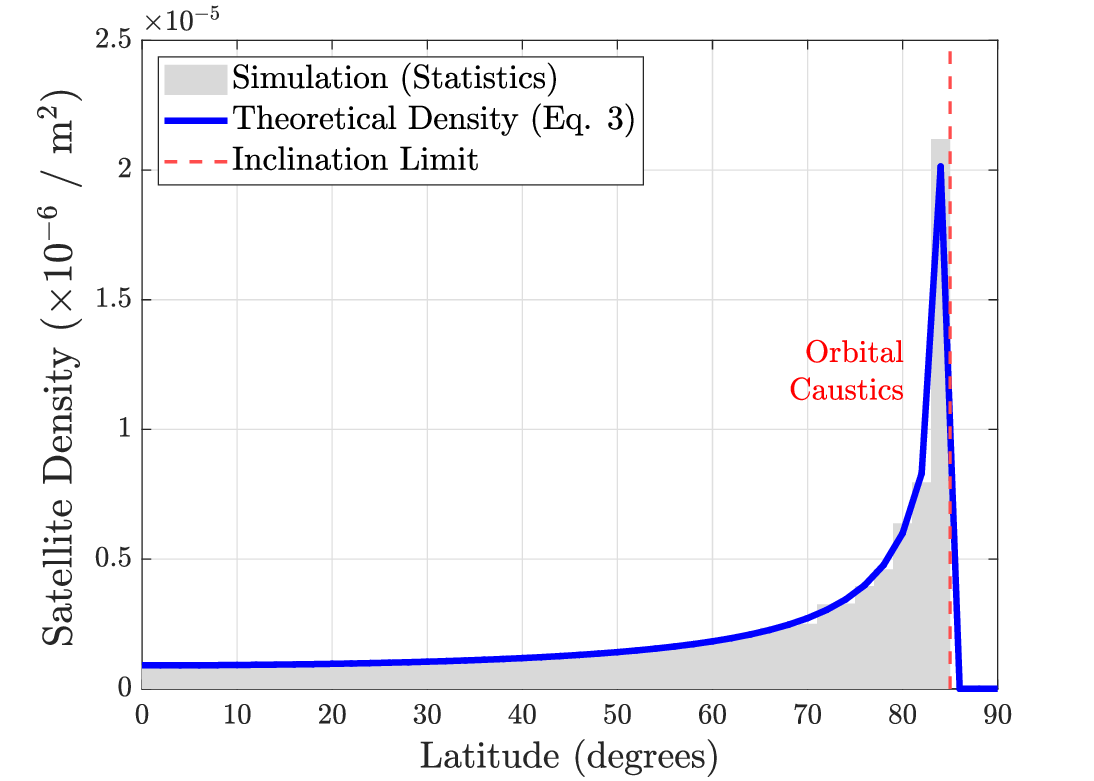}
			\vspace*{-5mm}
			\label{fig:compression_density} 
		}
		\subfigure[{\scriptsize Orbital compression effect}] {
			\includegraphics[width=0.8\linewidth]{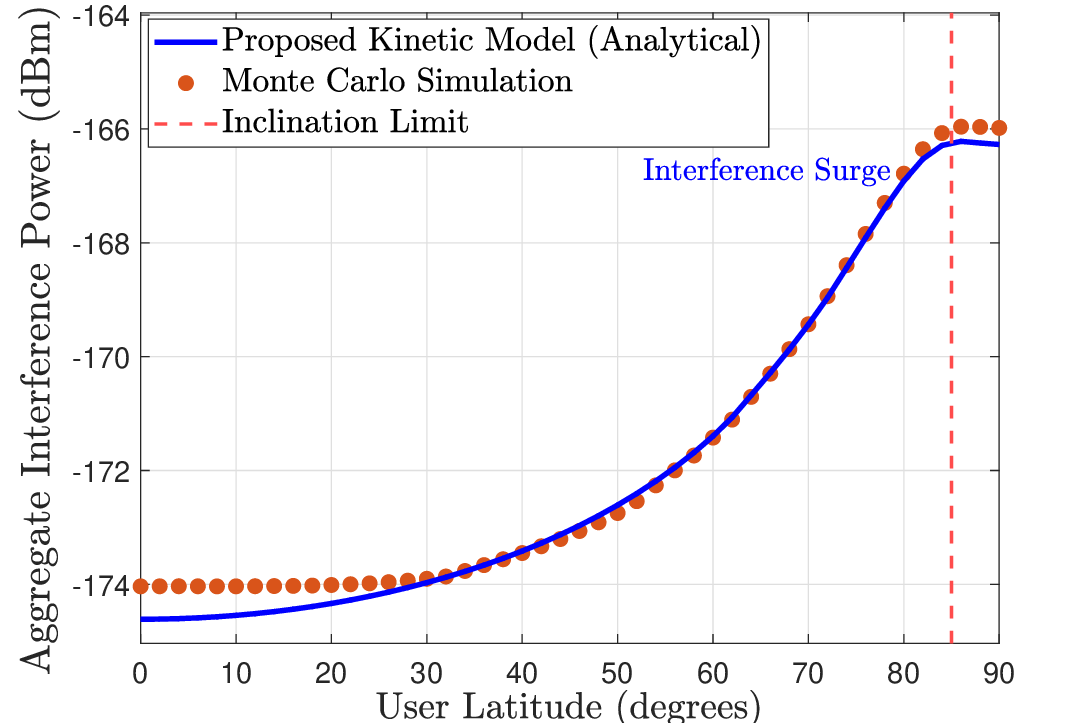}
			\label{fig:compression_power}
		}
		\vspace*{-3mm}
		\caption{Verification of orbital caustics.}
		\vspace*{-5mm}
		\label{fig:orbital_compression}
	\end{figure}
	
	\vspace*{-2mm}
	\subsection{Orbital Caustics and Dynamic Reliability}
	
	The spatial non-stationarity of the LEO constellation fundamentally alters the interference landscape across latitudes. Fig. \ref{fig:orbital_compression} provides a comparative analysis of this phenomenon. As predicted by the kinematic theory, Fig. \ref{fig:orbital_compression}(a) corroborates the existence of orbital caustics. The theoretical density profile exhibits an asymptotic surge near the inclination limit, closely matching the empirical histogram. Consequently, as shown in Fig. \ref{fig:orbital_compression}(b), the aggregate interference power rises monotonically with latitude, culminating in a deterministic surge near the caustic boundary. 
	
	\begin{figure}[t]
		\vspace*{-1mm}
		\centering
		\includegraphics[width=0.85\linewidth]{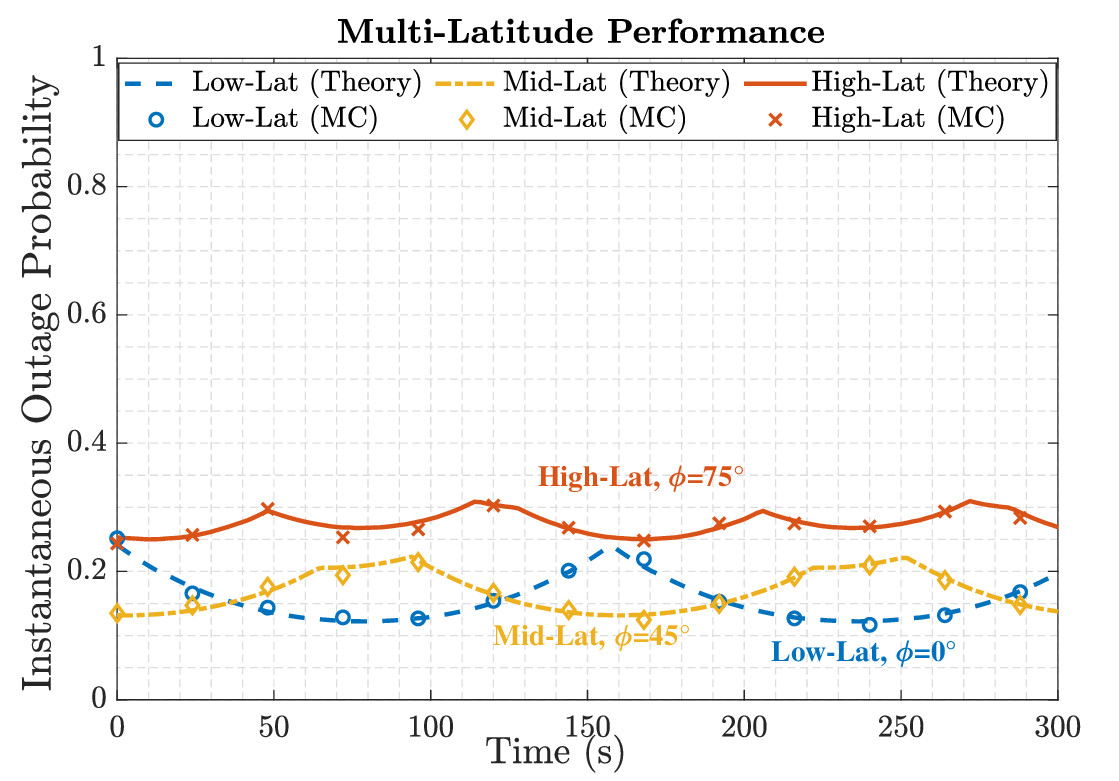} 
		\vspace*{-4mm}
		\caption{Comparison of time-varying outage probability at different latitudes.}
		\vspace*{-3mm}
		\label{fig:tvop_tradeoff}
	\end{figure}
	
	\begin{figure}[t]
		\vspace*{-1mm}
		\centering
		\includegraphics[width=0.8\linewidth]{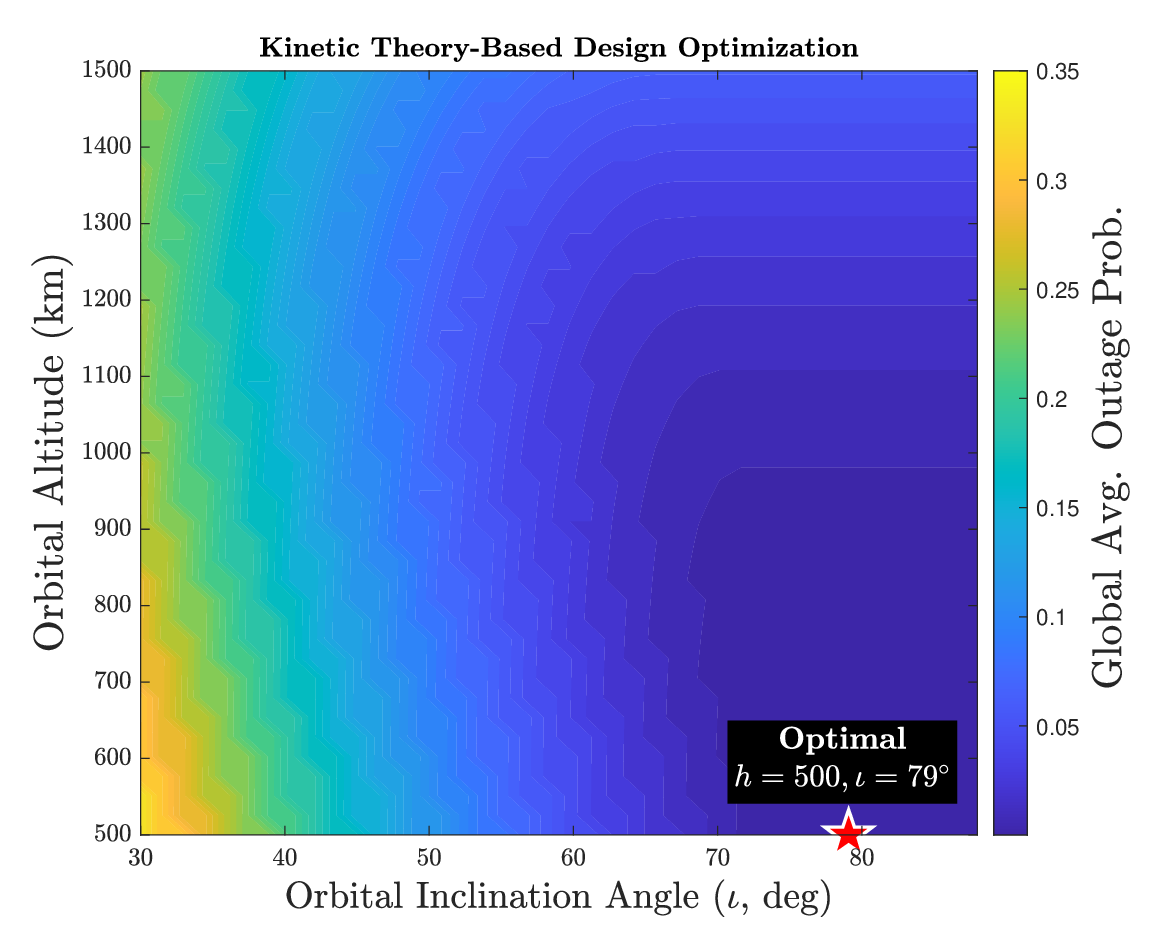} 
		\vspace*{-5mm}
		\caption{Global parameter optimization via kinetic theory.}
		\label{fig:optimal_design}
		\vspace*{-4mm}
	\end{figure}
	
	To quantify the direct impact of these orbital dynamics on link performance, Fig. \ref{fig:tvop_tradeoff} evaluates the TVOP across representative latitudes. The equatorial user exhibits the most robust reliability due to the reduced interference floor caused by orbital rarefaction. In contrast, the high-latitude scenario suffers significant deterioration despite enhanced satellite visibility. This confirms that the interference density scaling outpaces the geometric diversity gain. Furthermore, the asynchronous temporal fluctuations across latitudes underscore that aggregate interference is not mere background noise, but a structured spatiotemporal wave governed by the deterministic phasing of orbital planes.
	
	\vspace*{-2mm}
	\subsection{System-Level Topology Optimization}
	
	Exploiting the analytical tractability of the framework, we perform a global parameter optimization to minimize the average outage probability. The performance heat map in Fig. \ref{fig:optimal_design} reveals that system reliability is governed by a fundamental trade-off between coverage completeness and geometric interference isolation. Under the evaluated Ka-band regime and specified antenna patterns, the lowest average outage is obtained near an inclination of $79^{\circ}$, rather than the typically assumed $90^{\circ}$ polar orbit. 
	
	Physically, this angle represents the minimum sufficient inclination required to close the polar coverage gap at an altitude of 1000 km. Increasing the inclination beyond this point yields diminishing returns in coverage but severely exacerbates the orbital compression effect. While the exact optimal angle is inherently sensitive to specific payload parameters, the underlying mechanism-driven trend remains robust: the optimal strategy is to select the minimum inclination that satisfies global connectivity, thereby maximizing the spatial separation of orbital planes. This joint optimization provides a physics-based guideline for constellation design that balances capacity density with interference immunity.
	
	\vspace*{-2mm}
	\section{Conclusion}\label{S6}
	
	This paper presented a KIF framework to characterize the spatiotemporal interference dynamics in LEO mega-constellations. By mapping discrete orbital mechanics to a macroscopic fluid flow, we derived a hydrodynamic partial differential equation and closed-form outage probabilities that explicitly link physical orbital compression to high-latitude interference surges. Under the evaluated configuration, an inclination of approximately $79^{\circ}$ at low altitudes provides an effective trade-off between coverage continuity and geometric interference isolation, offering an alternative to a strict polar deployment. This continuum approach provides a scalable, physics-based methodology for evaluating and optimizing interference-limited non-terrestrial networks.
	
	\small

\end{document}